\documentclass[a4paper,USenglish,cleveref, autoref, thm-restate]{lipics-v2021}

\pdfoutput=1 
\hideLIPIcs  

\title{Computational Complexity of Swish} 

\funding{\emph{General}: JSPS KAKENHI Grant Numbers JP20H00605, JP24H00690, JP25H01114.}

\author{Takashi Horiyama}{Faculty of Information Science and Technology, Hokkaido University, Japan}{horiyama@ist.hokudai.ac.jp}{https://orcid.org/0000-0001-9451-259X}{JSPS KAKENHI Grant Number JP23K24806.}

\author{Takehiro Ito}{Graduate School of Information Sciences, Tohoku University, Japan}{takehiro@tohoku.ac.jp}{https://orcid.org/0000-0002-9912-6898}{JSPS KAKENHI Grant Number JP24H00686.}

\author{Jun Kawahara}{Graduate School of Informatics, Kyoto University, Japan \and \url{https://www.algo.cce.i.kyoto-u.ac.jp/jkawahara/}}{jkawahara@i.kyoto-u.ac.jp}{https://orcid.org/0000-0002-4096-3923}{}

\author{Shin-ichi Minato}{Graduate School of Informatics, Kyoto University, Japan}{minato@i.kyoto-u.ac.jp}{https://orcid.org/0000-0002-1397-1020}{}

\author{Akira Suzuki}{Center for Data-driven Science and Artificial Intelligence, Tohoku University, Japan \and \url{http://www.ecei.tohoku.ac.jp/alg/suzuki/}}{akira@tohoku.ac.jp}{https://orcid.org/0000-0002-5212-0202}{JSPS KAKENHI Grant Number JP25K14980.}

\author{Ryuhei Uehara}{School of Information Science, Japan Advanced Institute of Science and Technology, Japan \and \url{https://www.jaist.ac.jp/~uehara/}}{uehara@jaist.ac.jp}{https://orcid.org/0000-0003-0895-3765}{}

\author{Yutaro Yamaguchi}{Graduate School of Information Science and Technology, Osaka University, Japan \and \url{https://ygussany.github.io/index.html}}{yutaro.yamaguchi@ist.osaka-u.ac.jp}{}{JST CRONOS Japan Grant Number JPMJCS24K2.}

\authorrunning{T.\,Horiyama, T.\,Ito, J.\,Kawahara, S.\,Minato, A.\,Suzuki, R.\,Uehara, and Y.\,Yamaguchi} 

\Copyright{Jane Open Access and Joan R. Public} 



\category{} 

\relatedversion{} 

\nolinenumbers 

\EventEditors{John Q. Open and Joan R. Access}
\EventNoEds{2}
\EventLongTitle{42nd Conference on Very Important Topics (CVIT 2016)}
\EventShortTitle{CVIT 2016}
\EventAcronym{CVIT}
\EventYear{2016}
\EventDate{December 24--27, 2016}
\EventLocation{Little Whinging, United Kingdom}
\EventLogo{}
\SeriesVolume{42}
\ArticleNo{23}

\usepackage{mathtools}
\usepackage{cite}

\newcommand{\cC}{\mathcal{C}}
\newcommand{\cS}{\mathcal{S}}
\renewcommand{\rotate}[1]{{#1}^\mathrm{R}}
\newcommand{\flip}[1]{{#1}^\mathrm{F}}
\newcommand{\rotateflip}[1]{{#1}^\mathrm{RF}}

\newcommand{\problemdef}[3]{
    \begin{center}
    \fbox {   \parbox[c]{0.95\textwidth}{
        \textsc{\large #1} 
        
         \textbf{Input:} #2 \\
         \textbf{Goal:} #3 
        }}
    \end{center} 
}

\newcommand{\Swish}{\textsc{Swish}}

\usepackage{CJKutf8}

\AtBeginDocument{\begin{CJK}{UTF8}{min}}
\AtEndDocument{\end{CJK}}

\begin{document}

\maketitle

\begin{abstract}
Swish is a card game in which players are given cards having symbols (hoops and balls), and find a valid superposition of cards, called a ``swish.''
Dailly, Lafourcade, and Marcadet (FUN 2024) studied a generalized version of Swish and showed that the problem is solvable in polynomial time with one symbol per card, while it is NP-complete with three or more symbols per card.
In this paper, we resolve the previously open case of two symbols per card, which corresponds to the original game.
We show that Swish is NP-complete for this case.
Specifically, we prove the NP-hardness when the allowed transformations of cards are restricted to a single (horizontal or vertical) flip or 180-degree rotation, and extend the results to the original setting allowing all three transformations.
In contrast, when neither transformation is allowed, we present a polynomial-time algorithm.
Combining known and our results, we establish a complete characterization of the computational complexity of Swish with respect to both the number of symbols per card and the allowed transformations. 

\end{abstract}

\section{Introduction}
\label{sec:1_intro}
\Swish~\cite{Swish2011} is a card game in which players try to find a valid superposition of cards.
In the game, cards are overlaid so that every hoop (circle)~\textcircled{} meets a ball (point)~$\bullet$, every ball meets a hoop, and neither two hoops nor two balls overlap; such a configuration is called a \emph{swish}.
The commercial version of this game consists of $60$ transparent cards, each with exactly one hoop and exactly one ball arranged on a grid of height~$4$ and width~$3$ (see Figure~\ref{fig:swish} for examples).
Cards may be flipped horizontally or vertically and rotated by $180$ degrees, which yields $36$ distinct cards up to symmetry; some (equivalent) cards therefore appear twice in the deck.
In the commercial game, hoops and balls are colored in blue, orange, purple, and green, but these colors serve only to visually distinguish positions.
From a combinatorial point of view, the game can be regarded as \emph{monochromatic}, that is, symbols themselves carry no colors.
In each round of the game, given $16$ cards in common drawn from the deck, the players compete to find a swish faster than their opponents and call out ``Swish!''
Despite its simple rules, {\Swish} exhibits a rich combinatorial structure that makes it an appealing subject for mathematical and algorithmic analysis.

\begin{figure}
  \centering
  \includegraphics[width=0.9\linewidth]{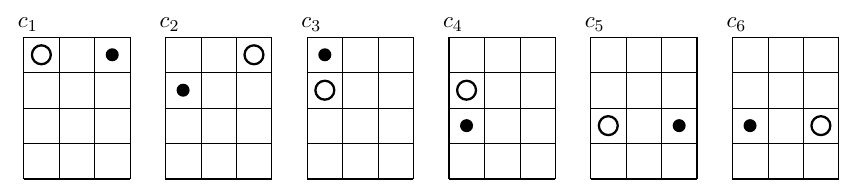}
  \caption{Example of cards in \Swish. Following the notation introduced in Section~\ref{sec:2_1_swish}, $c_1 = ((1, 1), (1, 3))$, $c_2 = ((1, 3), (2, 1))$, etc. The swishes of this instance without rotation or flip are $\{c_1, c_2, c_3\}$, $\{c_5, c_6\}$, and $\{c_1, c_2, c_3, c_5, c_6\}$.
  If we are allowed both flip and rotation, $\{c_2, c_3, c_6\}$ is also a swish, where we use the cards as $c_2 = ((1, 3), (2, 1))$, $\flip{c}_3 = ((2, 3), (1, 3))$, and $\rotate{c}_6 = ((2, 1), (2, 3))$, where the superscripts $\mathrm{F}$ and $\mathrm{R}$ mean horizontal flip and $180$-degree rortaion, respectively.}
  \label{fig:swish}
\end{figure}

Rowland~\cite{Rowland2017} initiated a mathematical study of {\Swish}, focusing on structural and probabilistic aspects of the commercial deck.
Using Burnside’s theorem, she showed that the $60$-card deck contains $36$ distinct cards up to symmetry.
She further analyzed the distribution of symbol positions, showing that they do not appear uniformly,
and studied the probability of forming a swish with two cards.
{\Swish} was also discussed from a linear-algebraic perspective via matrix representations.

From a computational complexity viewpoint, {\Swish} was first studied systematically by Dailly, Lafourcade, and Marcadet~\cite{dailly2024did} in FUN~2024.
They introduced a generalized version of {\Swish} and analyzed its complexity with respect to the number of symbols per card.
They showed that {\Swish} can be solved in polynomial time when each card contains a single symbol, whereas the problem becomes NP-complete when cards may contain three or more symbols.
To establish NP-completeness, they also considered a restricted variant in which cards cannot be flipped or rotated, often referred to as \textsc{Simple-Swish}.

In this paper, we complete the complexity classification of {\Swish} by addressing two aspects that remain open.
The first concern is the number of symbols per card: what is the complexity of {\Swish} when each card contains exactly two symbols, as in the original game?
The second concern is the role of geometric transformations: how does the complexity change when flip and rotation are selectively allowed or forbidden?

We show a sharp boundary.
When neither flip nor rotation is allowed, {\Swish} with two symbols per card admits a polynomial-time algorithm (Theorem~\ref{thm:main_poly}).
In contrast, if even a single type of transformation (i.e., one of the horizontal flip, vertical flip, and $180$-degree rotation) is allowed,
the problem becomes NP-complete, already with one hoop and one ball per card (Theorem~\ref{thm:main_hard}).
As a consequence, {\Swish} remains NP-complete in these settings even when cards may contain more symbols.
Together with previous results, this establishes a complete characterization of the computational complexity of {\Swish} with respect to both the number of symbols per card and the allowed symmetries. 

The rest of the paper is organized as follows.
In Section~\ref{sec:2_preliminaries}, we give definitions and formalize the problems.
In Section~\ref{sec:3}, we state our results and prove them separately in subsections.
In Section~\ref{sec:4_conclusion}, we conclude the paper with remarks on future work.
\section{Preliminaries}
\label{sec:2_preliminaries}

\subsection{Swish}\label{sec:2_1_swish}
We formalize a generalization of {\Swish} in terms of the size of cards.
Fix an $h \times w$ integer grid, where $(i, j)$ $(1 \le i \le h,\ 1 \le j \le w)$ denotes the point $i$-th from the top and $j$-th from the left; $h = 4$ and $w = 3$ in the original \Swish.
As we mainly focus on the case where each card has exactly one hoop and exactly one ball, each \emph{card} $c$ is specified as an ordered pair of two points on the grid, $((i_1, j_1), (i_2, j_2))$, which means that $c$ has a \emph{hoop} \textcircled{} at $(i_1, j_1)$ and a \emph{ball} $\bullet$ at $(i_2, j_2)$.
If we are allowed to \emph{rotate} a card $c = ((i_1, j_1), (i_2, j_2))$, we can use $c$ as a card $\rotate{c} \coloneqq ((h - i_1 + 1, w - j_1 + 1), (h - i_2 + 1, w - j_2 + 1))$; here, it is restricted to $180$-degree rotation even when $h = w$.
If we are allowed to \emph{flip} $c$, we can use $c$ as a card $\flip{c} \coloneqq ((i_1, w - j_1 + 1), (i_2, w - j_2 + 1))$; here, it is restricted to \emph{horizontal} flip by symmetry (because if both horizontal and vertical flips are allowed, their combination results in the $180$-degree rotation).
If we are allowed both to rotate and to flip $c$, we can use $c$ as, in addition to $\rotate{c}$ and $\flip{c}$, a card $\rotateflip{c} = ((h - i_1 + 1, j_1), (h - i_2 + 1, j_2))$, which coincides with the card obtained by flipping $c$ vertically.

Suppose that a set $\cC$ of cards is given.
A subset $\cS \subseteq \cC$ of cards is called a \emph{swish} if the following conditions are satisfied by flipping or/and rotating some of the cards in $\cS$ if they are allowed:
for every point $p$ on the grid, either
\begin{itemize}
    \item no card in $\cS$ has a hoop and no card in $\cS$ has a ball at $p$, or
    \item exactly one card in $\cS$ has a hoop and exactly one card in $\cS$ has a ball at $p$.
\end{itemize}
See Figure~\ref{fig:swish} again for examples.

The objective of the original game is to find a large swish.
Then, we can consider natural decision problems as follows:

\problemdef{Swish}
    {A set $\cC$ of cards on the $h \times w$ grid and an integer $k$.}
    {Decide whether there is a swish $\cS \subseteq \cC$ such that $|\cS| \ge k$ or not, where both flipping and rotating any cards are allowed.}

\problemdef{Swish without Flip}
    {A set $\cC$ of cards on the $h \times w$ grid and an integer $k$.}
    {Decide whether there is a swish $\cS \subseteq \cC$ such that $|\cS| \ge k$ or not, where rotating any cards is allowed but flipping them is not.}

\problemdef{Swish without Rotation}
    {A set $\cC$ of cards on the $h \times w$ grid and an integer $k$.}
    {Decide whether there is a swish $\cS \subseteq \cC$ such that $|\cS| \ge k$ or not, where flipping any cards is allowed but rotating them is not.}

\problemdef{Swish without Flip or Rotation}
    {A set $\cC$ of cards on the $h \times w$ grid and an integer $k$.}
    {Decide whether there is a swish $\cS \subseteq \cC$ such that $|\cS| \ge k$ or not, where neither flipping nor rotating any cards is allowed.}

\subsection{Graphs}
We use the standard terms and notation on graphs.
An undirected edge between two vertices $u$ and $v$ is denoted by an unordered pair $\{u, v\}$, and a directed edge from $u$ to $v$ is denoted by an ordered pair $(u, v)$.

In a graph (either undirected or directed), a \emph{matching} is an edge subset such that each vertex has at most one incident edge in it.
A matching is called \emph{perfect} if all the vertices have exactly one incident edge.
When each edge is assigned with a \emph{weight} value, the \emph{weight} of an edge subset is defined as the sum of the weights of edges therein.

In a directed graph, the \emph{in-degree} of vertex $v$ is the number of edges incoming to $v$, and the \emph{out-degree} is the number of edges outgoing from $v$.
A \emph{directed cycle} is a connected directed graph in which every vertex has both in-degree and out-degree exactly $1$.
The \emph{length} of a directed cycle is the number of vertices (or edges, equivalently) therein.
\section{Algorithm and Complexity}
\label{sec:3}

The main results of this paper, solving an open problem in \cite{dailly2024did}, are stated as follows:

\begin{theorem}\label{thm:main_poly}
    \textsc{Swish without Flip or Rotation} is solved in polynomial time.
\end{theorem}

\begin{theorem}\label{thm:main_hard}
    \textsc{Swish}, \textsc{Swish without Flip}, and \textsc{Swish without Rotation} are NP-complete.
\end{theorem}

Note that all the problems are clearly in NP, since if we are given a swish $\cS$ with additional information of how to use each card in $\cS$ (flipping or/and rotating it or not), we can easily confirm its validity in polynomial time.

\subsection{Polynomial-Time Algorithm for \textsc{Swish without Flip or Rotation}}
\label{sec:3_0axis}

In this section, we prove Theorem~\ref{thm:main_poly}.

Let $(\cC, k)$ be an instance of \textsc{Swish without Flip or Rotation}, i.e., $\cC$ is a set of cards and $k$ is an integer.
We construct from $\cC$ a bipartite graph $G = (V^+, V^-; E)$ with edge weight $w \colon E \to \mathbb{R}$ as follows (see Figure~\ref{fig:swish_to_graph}).

First, let $V$ be the set of points on the grid.
For each point $p \in V$, we create two copies $p^+$ and $p^-$.
Let $V^+ \coloneqq \{\, p^+ \mid p \in V \,\}$ and $V^- \coloneqq \{\, p^- \mid p \in V \,\}$.

Next, for each card $c = (p, q) \in \cC$, we create an edge $e_c = \{p^+, q^-\}$ of weight $w(e_c) = 1$.
Also, for each point $p \in V$, we create an edge $f_p = \{p^+, p^-\}$ of weight $w(f_p) = 0$.
Let $E$ be the set of created edges.

\begin{figure}
  \centering
  \includegraphics[width=0.33\linewidth]{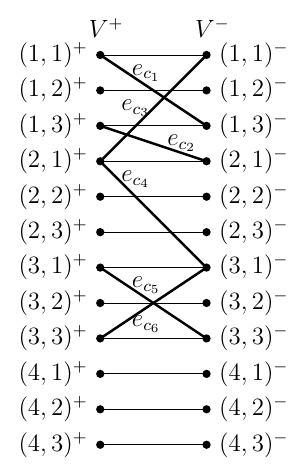}
  \caption{Bipartite graph $G = (V^+, V^-; E)$ constructed from the Swish instance in Figure~\ref{fig:swish}. $\{e_{c_1}, e_{c_2}, e_{c_3}\} \cup \{\, f_p = \{p^+, p^-\} \mid p \in V \setminus \{(1, 1), (1, 3), (2, 1)\} \,\}$ is a perfect matching, which corresponds to swish $\{c_1, c_2, c_3\}$ without rotation or flip. The weight of this perfect matching is $3$, which is equal to the size of the corresponding swish.}
  \label{fig:swish_to_graph}
\end{figure}

Since it is well-known that a maximum-weight perfect matching in a bipartite graph can be computed in polynomial time \cite{kuhn1955hungarian}, the following claim completes the proof.

\begin{claim}
    There is a swish $\cS \subseteq \cC$ such that $|S| = k$ without flipping or rotating any cards if and only if $(G, w)$ has a perfect matching of weight $k$.
\end{claim}

\begin{proof}
    Let $\cS \subseteq \cC$ be a swish such that $|S| = k$ without flipping or rotating any cards.
    We construct a perfect matching $M \subseteq E$ in $G$ with $w(M) = k$.
    For each card $c = (p, q) \in \cS$, we add the edge $e_c = \{p^+, q^-\}$ of weight $1$ to $M$.
    Then, as $\cS$ is a swish without flip or rotation, for each point $p \in V$, either
    \begin{itemize}
        \item both $p^+$ and $p^-$ have no incident edge of the current $M$, or
        \item both $p^+$ and $p^-$ have exactly one incident edge of the current $M$.
    \end{itemize}
    Thus, we can make $M$ a perfect matching by adding the edge $f_p$ of weight $0$ for each point $p$ in the former situation.
    The weight of the resulting perfect matching is clearly $|\cS| = k$.

    Conversely, let $M \subseteq E$ be a perfect matching in $G$ with $w(M) = k$.
    We construct a swish $\cS \subseteq \cC$ such that $|S| = k$ without flipping or rotating any cards.
    Let $\cS \coloneqq \{\, c \in \cC \mid e_c \in M \,\}$.
    By definition of edge weight $w$, we have $|\cS| = w(M) = k$.
    Since $M$ is a perfect matching and each edge $f_p \in M$ connects two vertices $p^+$ and $p^-$ from the same point $p \in V$, if one of $p^+$ and $p^-$ is covered by some edge $e_c \in M$, then the other is also covered by some edge $e_{c'} \in M$.
    This implies that for each point $p \in V$, either
    \begin{itemize}
        \item no card in $\cS$ has a hoop and no card in $\cS$ has a ball at $p$, or
        \item exactly one card in $\cS$ has a hoop and exactly one card in $\cS$ has a ball at $p$.
    \end{itemize}
    Thus, $\cS$ is a swish without flip or rotation, which completes the proof.
\end{proof}

\begin{remark}
    The computational time is bounded by $\mathrm{O}(n^3)$ or $\mathrm{O}((m + n)n \log n)$ \cite{kuhn1955hungarian,edmonds1972theoretical}, where $n \coloneqq |V| = hw$ and $m$ is the number of input cards.
    Here, we may assume that $h$ and $w$ are bounded by $\mathrm{O}(m)$, since a row or column that contains neither hoop nor ball can be removed without loss of generality.
    Thus, this is polynomial in the input size.
\end{remark}

\begin{remark}
    We can consider more general settings of \textsc{Swish without Flip or Rotation}.
    Suppose that we want to find a swish of size exactly $k$.
    Then, the same reduction results in the so-called \emph{exact matching} problem, in which we are required to find a perfect matching of weight \emph{exactly} $k$.
    For this problem, a \emph{randomized} polynomial-time algorithm \cite{mulmuley1987matching} is well-known, but the existence of a \emph{deterministic} polynomial-time algorithm has been a major open problem in theoretical computer science for more than 40 years since its appearance \cite{papadimitriou1982complexity}.
    
    As another generalization, suppose that each input card may have either no hoop but exactly two balls at different points or no ball but exactly two hoops at different points.
    Also in this case, the reduction and the proof are almost the same (e.g., if a card $c$ has two hoops at points $p$ and $q$, then we create an edge $e_c = \{p^+, q^+\}$ of weight $w(e_c) = 1$), but the resulting graph is no longer bipartite.
    Even though, it is known that a maximum-weight perfect matching can be computed in polynomial time (deterministic) \cite{edmonds1965maximum}, and that an exact-weight perfect matching can be computed in polynomial time (randomized) \cite{mulmuley1987matching}.
    In both cases, it is known that the decision problem can be solved in $\mathrm{O}(n^\omega)$ time (randomized) \cite{sankowski2009maximum, huang2012efficient, sato2025exact} and an extra $n$ factor is required to find a solution itself in the latter (exact-weight) case, where $\omega \le 2.371339$ is the matrix multiplication exponent \cite{alman2025more}.
\end{remark}
\subsection{NP-hardness of \textsc{Swish without Rotation}}
\label{sec:3_1axis}

In this section, we prove that \textsc{Swish without Rotation} is NP-hard.
The proof of Theorem~\ref{thm:main_hard} will be completed based on this proof in the next section.

We reduce the following problem, which is known to be NP-complete.
For a directed graph $G = (V, E)$, we call an edge subset $F \subseteq E$ an \emph{even dicycle-factor} if $F$ covers all vertices of $V$ and each connected component of $F$ is a directed cycle of even length (see Figure~\ref{fig:even_cycle_factor}).

\problemdef{$\forall$Even Dicycle-Factor}
    {A directed graph $G = (V, E)$.}
    {Decide whether there is an even dicycle-factor $F \subseteq E$ or not.}

\begin{figure}
  \centering
  \begin{minipage}[b]{0.48\columnwidth}
  \includegraphics[width=0.9\linewidth]{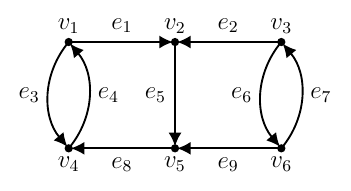}
  \end{minipage}
  \begin{minipage}[b]{0.48\columnwidth}
  \includegraphics[width=0.9\linewidth]{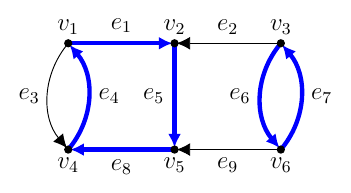}
  \end{minipage}
  \caption{Example of an instance (left) and its solution (right) of \textsc{$\forall$Even Dicycle-Factor}.}
  \label{fig:even_cycle_factor}
\end{figure}

\begin{theorem}[\!\!\cite{bang2014arc,horsch2025odd}]\label{thm:even_dicycle_factor}
    \textsc{$\forall$Even Dicycle-Factor} is NP-complete.
\end{theorem}

Let $G = (V, E)$ be an instance of \textsc{$\forall$Even Dicycle-Factor}, i.e., a directed graph.
We construct an equivalent instance $(\cC, k)$ of \textsc{Swish without Rotation} as follows.

Let $h \coloneqq |V|$ and $w \coloneqq 4$.
We index the rows by $V$ and the columns by $\{1, 2, 3, 4\}$, and simply denote by $v^j$ each point $(v, j)$ on the grid.

For each edge $e = (u, v) \in E$, we create a card $c_e = (u^1, v^1)$.
In addition, for each vertex $v \in V$, we create three cards $f_{v,j} = (v^j, v^{j+1})$ $(j = 1, 2, 3)$.
Let $\cC$ be the set of created cards (see Figure~\ref{fig:graph_to_card}).

\begin{figure}
  \centering
  \includegraphics[width=0.9\linewidth]{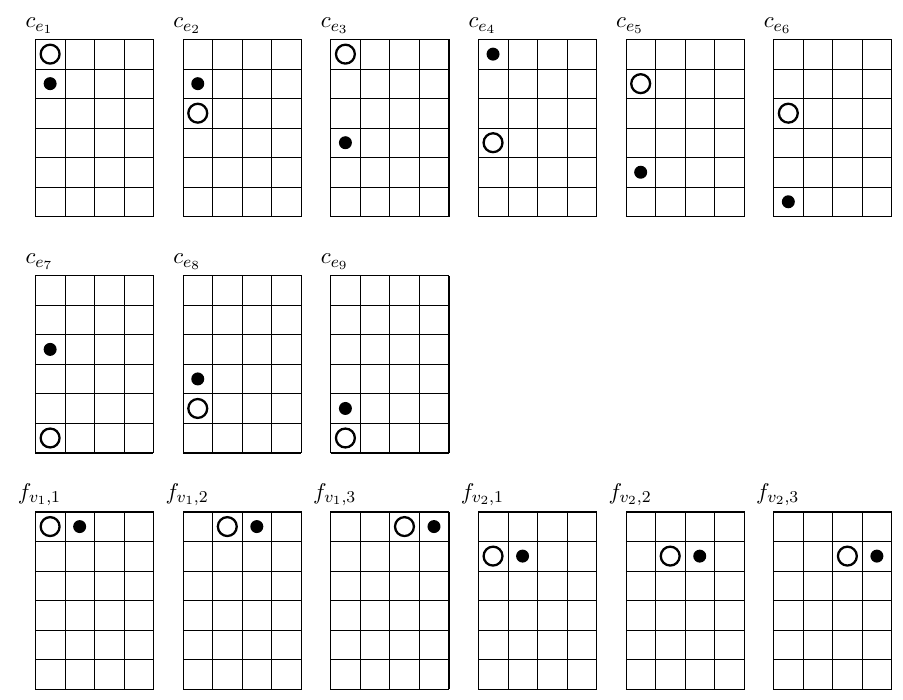}
  \caption{The set of cards created from the instance of \textsc{$\forall$Even Dicycle-Factor} in Figure~\ref{fig:even_cycle_factor}, where we omit $f_{v_i, j}$ for $i = 3, 4, 5, 6$ (they are analogous to $f_{v_1, j}$ and $f_{v_2, j}$).}
  \label{fig:graph_to_card}
\end{figure}

\begin{figure}
  \centering
  \includegraphics[width=0.9\linewidth]{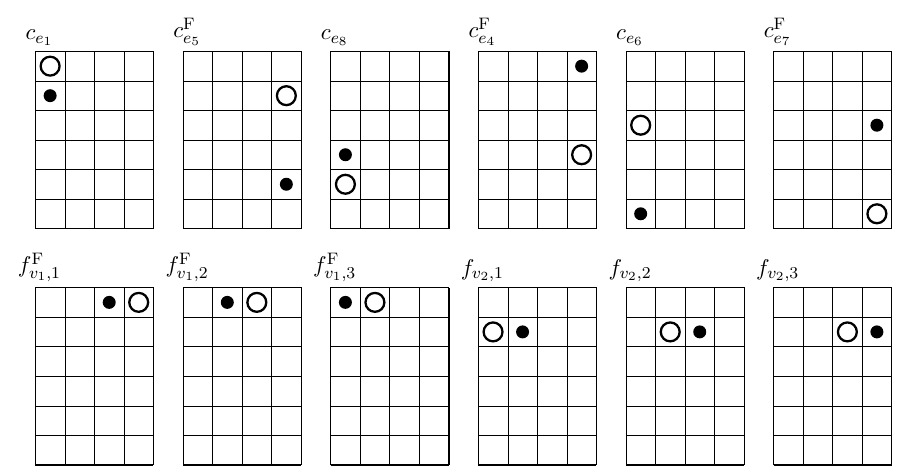}
  \caption{The swish without rotation corresponding to the solution in Figure~\ref{fig:even_cycle_factor}, where we omit $f_{v_i, j}$ for $i = 3, 4, 5, 6$.}
  \label{fig:swish_flip}
\end{figure}

The following claim completes the reduction, where we set $k = 4|V|$ (see Figure~\ref{fig:swish_flip}).

\begin{claim}
    $G$ has an even dicycle-factor if and only if $\cC$ admits a swish of size $4|V|$, where flipping cards is allowed but rotating cards is forbidden.
\end{claim}

\begin{proof}
    Let $F \subseteq E$ be an even dicycle-factor of $G$.
    Since each connected component of $F$ is a directed cycle of even length, we can color the edges of $F$ using two colors, say \emph{red} and \emph{blue}, such that any pair of adjacent edges receives different colors.
    Fix such a coloring, and let $\cS \coloneqq \{\, c_e \mid e \in F \,\} \cup \{\, f_{v, j} \mid v \in V,\ j \in \{1, 2, 3\} \,\}$.
    Note that $|\cS| = |F| + 3|V| = 4|V|$.
    We prove that $\cS$ is a swish.

    First, we fix how to use the card $c_e = (u^1, v^1) \in \cS$ for each $e = (u, v) \in F$.
    If $e$ is red, we use it as it is, i.e., as $c_e = (u^1, v^1)$.
    In contrast, if $e$ is blue, we use it by flipping, i.e., as $\flip{c}_e = (u^4, v^4)$.
    As $F$ is an (even) dicycle-factor, each vertex has exactly one incoming edge of $F$ and exactly one outgoing edge of $F$, which are colored differently.
    Thus, following the above definition, for each vertex $v \in V$, one of the two points $v^1$ and $v^4$ receives exactly one hoop and no ball and the other receives exactly one ball and no hoop.
    If $v^1$ receives a ball (and $v^4$ receives a hoop), then we use $f_{v,j} = (v^j, v^{j+1}) \in \cS$ $(j = 1, 2, 3)$ as they are, and otherwise, we use them by flipping, i.e., as $\flip{f}_{v,j} = (v^{5-j}, v^{4-j})$.
    In both cases, the two middle points $v^2$ and $v^3$ receive exactly one hoop and exactly one ball from these three cards, and $v^1$ and $v^4$ receive exactly one of a hoop and a ball that is consistent with what is received from the cards $c_e \in \cS$.
    Thus, $\cS$ is indeed a swish without rotation.

    We now prove the reverse direction.
    Let $\cS$ be a swish of size $4|V|$ without rotation, which means that we can put exactly one hoop and exactly one ball at every point on the grid only using flip.
    First, the middle points $v^2$ and $v^3$ for each vertex $v \in V$ can be covered only by the three cards $f_{v,j}$ $(j = 1, 2, 3)$, so we must use them.
    In particular, due to the constraint on the middle points, we can use the three cards either as they are or by flipping all of them.
    Then, we use exactly $|V|$ cards in the form of $c_e = (u^1, v^1)$ for edges $e = (u, v) \in E$ in total such that for each vertex $v \in V$, one of $v^1$ and $v^4$ receives exactly one hoop and no ball and the other receives exactly one ball and no hoop.

    Let $F \coloneqq \{\, e \in E \mid c_e \in \cS \,\}$.
    By the above argument, for each vertex $v \in V$, exactly one incoming edge (corresponding to a ball) and exactly one outgoing edge (corresponding to a hoop) is contained in $F$; this means that every connected component of $F$ is a directed cycle, which covers all the vertices.
    In addition, one of the two corresponding cards must be used as it is and the other must be flipped.
    According to this, we can color each directed cycle in $F$ using two colors (flipped or not) such that any pair of adjacent edges receives different colors, which implies that the length must be even.
    Thus, $F$ is an even dicycle-factor, which completes the proof.
\end{proof}

\begin{remark}
It was not explicitly written in \cite{bang2014arc, horsch2025odd} but the proofs of Theorem~\ref{thm:even_dicycle_factor} imply that \textsc{$\forall$Even Dicycle-Factor} is NP-hard even when the input graph is restricted so that the following conditions are both satisfied:
\begin{itemize}
    \item every vertex has in-degree and out-degree at most $3$ (by \cite{horsch2025odd} immediately and also by \cite{bang2014arc} with preprocessing the input graph before the reduction, for the $2$-linkage problem), and
    \item it is planar (by \cite{horsch2025odd} combined with \cite{marx2005np}).
\end{itemize}
Furthermore, it is easy to see that we can strengthen the first restriction on the input graph $G = (V, E)$ as follows:
\begin{itemize}
    \item $G$ is \emph{tripartite}, where $G$ admits a partition $(V^-, V^\circ, V^+)$ of the vertex set $V$ and every edge in $E$ is either from $V^-$ to $V^\circ$, from $V^\circ$ to $V^+$, or from $V^+$ to $V^-$,
    \item every vertex in $V^-$ has out-degree $1$ and in-degree at most $3$,
    \item every vertex in $V^\circ$ has both in-degree and out-degree $1$, and
    \item every vertex in $V^+$ has in-degree $1$ and out-degree at most $3$.
\end{itemize}
This can be observed via a standard transformation of directed graphs: split each vertex $v$ into three copies $v^-$ (for incoming edges), $v^\circ$ (for preserving the parity of length), and $v^+$ (for outgoing edges) with two extra edges $(v^-, v^\circ)$ and $(v^\circ, v^+)$, and replace each original edge $(u, v)$ with $(u^+, v^-)$.
This transformation does not change the problem.

Based on this restriction on the input graph $G = (V, E)$ of \textsc{$\forall$Even Dicycle-Factor}, we can impose the following restriction on the input set $\cC$ of cards of \textsc{Swish without Rotation}:
$(*)$ for each point $p$ on the grid, either
\begin{itemize}
    \item exactly one card has a hoop and at most two cards have a ball at $p$, or
    \item exactly one card has a ball and at most two cards have a hoop at $p$.
\end{itemize}
This can be seen by constructing an equivalent set $\cC'$ of cards from the set $\cC$ constructed in the proof.

By the tripartition and degree conditions, the edges from $V^-$ to $V^\circ$ and those from $V^\circ$ to $V^+$ form matchings.
For such edges $e = (v^-, v^\circ)$ and $e' = (v^\circ, v^+)$ $(v^- \in V^-,\ v^\circ \in V^\circ,\ v^+ \in V^+)$, we include the cards $c_{e} = ((v^-)^1, (v^\circ)^1)$ and $c_{e'} = ((v^\circ)^1, (v^+)^1)$ in $\cC'$ as in $\cC$.
In addition, we may assume that there exists a perfect matching between $V^+$ and $V^-$ consisting of edges from $V^+$ to $V^-$, because otherwise $G$ is clearly a no-instance of \textsc{$\forall$Even Dicycle-Factor} ($G$ does not admit any cycle-factors even without parity constraint), which can be tested in polynomial time.
Fix any such perfect matching $M$, and for each edge $e = (u^+, v^-) \in M$, we include the card $c_e = ((u^+)^1, (v^-)^1)$ in $\cC'$ as in $\cC$.
For each of the remaining edges $e = (u^+, v^-) \in E \setminus M$ from $V^+$ to $V^-$, we include the flipped card $\flip{c}_e = ((u^+)^4, (v^-)^4)$ in $\cC'$ instead of the card $c_e \in \cC$.
By construction, at this point, $\cC'$ satisfies the following:
\begin{itemize}
    \item for each vertex $v \in V = V^- \cup V^\circ \cup V^+$, exactly one card has a hoop and exactly one card has a ball at $v^1$;
    \item for each vertex $v^- \in V^-$, at most two cards have a ball and no card has a hoop at $(v^-)^4$;
    \item for each vertex $v^\circ \in V^\circ$, any card has neither a hoop nor a ball at $(v^\circ)^4$;
    \item for each vertex $v^+ \in V^+$, at most two cards have a hoop and no card has a ball at $(v^+)^4$.
\end{itemize}
In other words, the desired condition $(*)$ will be satisfied if each $v^1$ receives at most one of a hoop and a ball, each $(v^-)^4$ receives exactly one hoop, each $(v^\circ)^4$ receives exactly one of a hoop and a ball, and each $(v^+)^4$ receives exactly one ball, respectively, in addition.
It can be indeed achieved as follows: for each $v \in V^- \cup V^\circ$, we include the flipped cards $\flip{f}_{v,j} = (v^{5-j}, v^{4-j})$ $(j = 1, 2, 3)$ in $\cC'$ instead of $f_{v,j} \in \cC$, and for each $v^+ \in V^+$, we include the cards $f_{v^+,j} = ((v^+)^j, (v^+)^{j+1})$ $(j = 1, 2, 3)$ in $\cC'$ as in $\cC$.
Note that each of the middle points $v^2, v^3$ receives exactly one hoop and exactly one ball in either case.
Thus, we can construct a set $\cC'$ of cards that is equivalent to $\cC$ and satisfies the condition $(*)$.
\end{remark}
\subsection{NP-hardness of \textsc{Swish}}
\label{sec:3_2axis}
In this section, we prove that \textsc{Swish} and \textsc{Swish without Flip} are NP-hard as well as \textsc{Swish without Rotation} shown in the previous section.
The proof is almost the same, and the only difference is the coordination of points on the grid.

In the previous proof, we used the $|V| \times 4$ grid, where each point $(v, j)$ $(v \in V,\ 1 \le j \le 4)$ is denoted simply by $v^j$.
Here we consider arranging these points on a single vertical line.
Specifically, let $h \coloneqq 4|V|$ and $w \coloneqq 1$, and name the $4|V|$ points on the $h \times w$ grid using $v^j$ $(v \in V,\ 1 \le j \le 4)$ so that the point $(i, 1)$ is named $v^j$ if and only if the point $(h - i + 1, 1)$ is named $v^{5-j}$.
Then, the horizontal flip becomes meaningless, and the rotation plays the same role as the horizontal flip in the previous section.
Thus, the same reduction works for this setting (rotation is allowed and horizontal flip is either allowed or not), which completes the proof of Theorem~\ref{thm:main_hard}.
\section{Concluding Remarks}
\label{sec:4_conclusion}

In this paper, we have completed the complexity classification of {\Swish} by resolving the previously open case of two symbols per card.
Our results show that the computational complexity of {\Swish} is under interplay between the number of symbols per card and the geometric transformations allowed on cards.
In particular, we have identified a sharp boundary: when neither flip nor rotation is allowed, {\Swish} with two symbols per card can be solved in polynomial time, whereas allowing even a single type of flip or $180$-degree rotation already makes the problem NP-complete.
Together with earlier results, this establishes a complete picture of the complexity of {\Swish} across all natural variants of the game.

Beyond decision complexity, an intriguing line of research concerns the structure of card sets that admit no swish.
Dailly, Lafourcade, and Marcadet~\cite{dailly2024did} initiated a systematic study of such \emph{no-swish} configurations.
In their generalized setting, they constructed large families of card sets containing no swish, achieving asymptotically high densities.
Moreover, for the original commercial deck of $60$ cards, they identified a no-swish subset of size $28$.

A natural and still open question is whether this value $28$ is optimal.
Is it possible to find a larger no-swish subset within the original {\Swish} deck, or does every set of $29$ cards necessarily admit a swish?
More broadly, determining tight bounds on the size of no-swish sets for fixed decks or restricted symmetry models appears to be a promising direction for future work, which would deepen our understanding of the combinatorial structure underlying the game and could require techniques different from those used in the generalized setting.

\bibliography{ref_swish}
\end{document}